\documentclass[11pt]{article}

\date{}

\usepackage{url}
\usepackage{cite}
\usepackage{plain}

\usepackage{wrapfig}
\usepackage{graphics}
\usepackage[english]{babel}
\usepackage[leqno]{amsmath}
\usepackage{amssymb}
\usepackage{verbatim}
\usepackage[mathscr]{eucal}

\usepackage{amstext}
\usepackage{makeidx}
\usepackage{amsthm}

\usepackage{hyperref}
\hypersetup{colorlinks,%
citecolor=red,%
linkcolor=blue,%
urlcolor=black}

\makeindex

\newtheorem{theorem}{Teorema} 

\newtheorem{proposition}{Proprieta'}
\newtheorem{definition}{Definizione}
\newtheorem{notation}{Nota}
\newtheorem{ex}{Esercizio} 
\newtheorem{esempio}{Esempio}

\newcommand{\beq}{\begin{equation}} 
\newcommand{\eeq}{\end{equation}}

\newcommand{\bex}{\begin{ex}} 
\newcommand{\eex}{\end{ex}} 

\newcommand{\bese}{\begin{esempio}} 
\newcommand{\eese}{\end{esempio}} 

\newcommand{\bpro}{\begin{proposition}} 
\newcommand{\epro}{\end{proposition}}

\newcommand{\bthe}{\begin{theorem}} 
\newcommand{\ethe}{\end{theorem}}

\newcommand{\bnote}{\begin{notation}} 
\newcommand{\enote}{\end{notation}}

\newcommand{\bdefi}{\begin{definition}} 
\newcommand{\edefi}{\end{definition}} 

\newcommand{\bc}{\begin{center}} 
\newcommand{\ec}{\end{center}}

\newcommand{\mail}[1]{\href{unina:#1}{\texttt{#1}}}

\usepackage{pdfsync}

\author{Monica De Angelis\thanks{Univ. of Naples  "Federico II", Dip. Mat. Appl. "R.Caccioppoli",
 \newline\mail{modeange@unina.it}}}

\title{Hopf bifurcations in dynamics of excitable systems}

\begin{document}
\maketitle

\abstract{A general FitzHugh-Rinzel model, able to describe several  neuronal phenomena, is considered. Linear stability and Hopf bifurcations are investigated by means of the spectral equation for the ternary  autonomous   dynamical system and the analysis is driven by both an admissible critical point and a  parameter which characterizes the system.}

\vspace{3mm} \textbf{Keywords} {FitzHugh Rinzel model,  Linear stability, Hopf bifurcations,  Neuron bursting frequency}


\maketitle

\section{Introduction }

The physiological and chemical properties that characterize neurons make them able to receive, process and transmit electrical signals that, associated with ionic currents, cross the membrane of the neuron. These electrical signals are called nerve impulses, while the difference in electrical charge that exists between the inside and outside of the neuronal cell is called membrane potential. The variation in the membrane potential is called action potential and it travels along the axon and is transmitted unchanged to other neurons in the form of electrical impulses. In this way, information is transmitted from one neuron to another, forming what is known as  synapse. This phenomenon is well known in literature and an extensive bibliography exists in regard \cite{j62,i,ks}. A reference point for these studies are   the works  of Hodgkin and Huxley [HH], who developed the model of the propagation of an electrical signal along a squid axon  (an axon so great to be called  giant). Their model consists of a system of four differential equations describing the dynamics of the membrane potential and the three fundamental ionic currents: the sodium current, the potassium current and the leakage current, which is mainly due to chlorine but also considers the effect of other minor ionic currents.  However, the non linearity and high dimensionality of the HH model made the analysis too complicated, so that simpler models were introduced to allow the essential aspects of the dynamics of models to be captured.

One of these models is the FitzHugh-Nagumo system (FHN) where, indicating by  $ \, U(x,t)\, $  the trasmembrane potential and by $\,W(x,t)\,$  a  variable associated with the contributions to the membrane current from sodium, potassium and other ions, it is given by

\begin{equation}
\label{FHN}
  \left \{
   \begin{array}{lll}
    \displaystyle{\frac{\partial \,U }{\partial \,t }} =\,  D \,\frac{\partial^2 \,U }{\partial \,x^2 }
     \,-\, W\,\, + f(U ) \,  \\
\\
\displaystyle{\frac{\partial \,W }{\partial \,t } }\, = \, \varepsilon (-\beta W +c +U).
  \end{array}
  \right.
\end{equation}

\vspace{3mm}
 \noindent Constant $ D\, > \, 0\, $ is a diffusion coefficient related to the  axial current in the axon. It follows from the HH theory where, denoting by $ d $ the diameter of the axon and  by $ r_i $ the resistivity, the spatial variation of the potential $ V $ gives the term $ (d/4 r_{i}) V_{xx} $ from which the term $ D\, U_{xx} $  is deduced \cite{j62}. Furthermore $ \varepsilon,  \, c, \,$ and $\, \beta \, $ are  constants that characterize the model's kinetic.

 The documentation is numerous and the analysis is extensive (see, for instance,    \cite{krs,mda18} and references therein).

\noindent As for  function $\, f(U),\,$  it depends on the reaction kinetics of the model and can  assume various expressions such as  a piecewise linear form,   or  $ f(U) = U-U^3/3 $. Besides, more in general,  function $ f(u)  $ assumes the following form \cite{i,ks}:

\begin{equation}                 \label{12}
f(U)= U\, (\, a-U \,) \, (\,U-1\,). \,
\end{equation}

\vspace{3mm}\noindent The cubic term is due to an instantaneous inversion of the sodium permeability and can be thought to play the same role as the $ m $  variable in the HH model, where the variable of activation of the channels of sodium is considered. Hence,  $ a  $  represents a threshold constant  and  is  an excitability parameter\cite{GAR}. In addition,  $ a $ can take both positive and negative values (see,f.i.\cite{zb}) and cases  with function $ a(x) $ are considered in \cite{AD} for inhomogeneous means.

Besides, one aspect worth noting is the existence of an equivalence between the FHN model and the third-order equation characterizing Josephson junctions in superconductivity {\cite{df213,nono,scott}}.  It follows that the analysis of such models is reflected in both biological and superconducting phenomena and, in addition, in dissipative problems \cite{mrvoigt,mda12,ddf12}.

Similarly, in order to investigate other phenomena such as, for example, bursting oscillations, the  well known system of FitzHugh-Rinzel (FHR) can be considered \cite{k2019,kl,arxiv,misa, moca21}.
This model is derived from the FHN model  and, unlike the latter, has an additional variable that changes periodically from a rapid spike oscillation to a silent phase during which the membrane potential changes slowly \cite{ks}.

Indeed, bursting phenomena occur in various scientific fields (see, f.i.\cite{2020} and references therein), and many devices are being built to mimic the behavior of a biological synapse, suggesting that electronic synapses may be introduced in the future to directly connect neurons \cite{clag}.  As a result, the FHR system is increasingly being studied to provide a mathematical description of  physical phenomena occurring in organisms.

\vspace{3mm} The  FitzHugh-Rinzel   model  considered  in this paper is the   following one:

\begin{equation}
\label{22}
  \left \{
   \begin{array}{lll}
    \displaystyle{\frac{d \,U }{d \,t }} =\,   -a\,U \,\,+ U^2\, (\,a+1\,-\frac{1}{k}\,U\,)
     \,-\, W\,\,+Y\,\,  +I \,  \\
\\
\displaystyle{\frac{d \,W }{d \,t } }\, = \, \varepsilon (-\beta W +c +U)
\\
\\
\displaystyle{\frac{d \,Y }{d \,t } }\, = \,\delta (-U +h -dY)
   \end{array}
  \right.
\end{equation}


\vspace{3mm}\noindent  where the physical variables $(U,W,Y)$ represent, respectively, the transmembrane potential, the recovery variable
\noindent and the slow current in the dendrite. Moreover, the parameter $ \varepsilon $ specifies the relationship between the time constants of the activator and inhibitor \cite{GAR}, and  $ c $ and $ \beta $ can be related to the number of cell membrane channels open to sodium and potassium ions, respectively\cite{rr}. Constant  $ I $ measures the amplitude of the external stimulus current and is modulated by the variable $ Y $ on a slower time scale \cite{ks}. In addition, if $ \beta \varepsilon $ and $ \delta d $ are positive constants, they can be regarded as the coefficients of viscosity \cite{R3}.


 \vspace{3mm} When  $ k=3 $ and $ a=-1, $ (\ref{22}) turns into this model:

\begin{equation}
\label{21p}
  \left \{
   \begin{array}{lll}
    \displaystyle{\frac{d \,U }{d \,t }} =\, U-U^3/3 + I 
     \,-\, W\,\,+Y\,\, \,  \\
\\
\displaystyle{\frac{d \,W }{d \,t } }\, = \, \varepsilon (-\beta W +c +U)
\\
\\
\displaystyle{\frac{d \,Y }{d \,t } }\, = \,\delta (-U +h -dY).
   \end{array}
  \right.
\end{equation}                                                      

\vspace{3mm}\noindent often studied  in literature (see, f.i. \cite{kl,R3,k2019} and references therein).

\vspace{3mm} Aim of the paper is to analyze the linear stability of the critical points of the FHR system, as well as to highlight the cases of Hopf bifurcations. Considering the spectrum equation, and its eigenvalues,  stability is evaluated by the Lienard-Chipart criterion. Furthermore, for what concerns instability, showing that  the problem can be expressed by way of  a positive parameter $ R $, the steady and/or oscillatory Hopf  bifurcations cases are determined by means of the instability coefficient power (ICP) method introduced by Rionero (see, f.i. \cite{R3,R4} and references therein). The plan of the paper is the following one: section 2 highlights some premises by which the subsequent theorems will be proved. In section  3 the mathematical problem and  linear operator $ L  $  with its invariants  is given. Finally, in section 4 and section 5, Hopf bifurcations driven by critical point $ \bar U  $  and driven by coefficient  $ -\eta= - \varepsilon  \beta  $ are evaluated.

\section{Some premises}

Due to the oscillatory activities of neurons, the onset of oscillatory bifurcations has gained the attention of many researchers.  Regarding the study of Hopf's bifurcations,  an extensive literature exists (see, f.i. \cite{R1,ccdt,R3,R4} and references therein). In order to justify the results stated here, some introductory considerations will be required. 

Indeed, in relation to linear stability, according to \cite{R1} when a phenomenon is modelled by  the system: 

\begin{equation}
\frac{d\textbf{U}}{dt} = \textbf{F} \quad \quad   t\geq 0, \quad  \textbf{U}(0)= \textbf{U}_0
\end{equation}


\noindent   introducing a fixed solution  $ \bold{\bar U} $ and the perturbation 
$ \textbf{u} = \textbf{U}- \bf{\bar U},  $ the behaviour of $ \textbf{u} $ is governed by:

\begin{equation} \label{2}
\frac{d \textbf{u}}{dt} = L \textbf{u} + N\textbf{u}, \quad \quad   t\geq 0, \quad  \textbf{u}(0)= \textbf{u}_0
\end{equation}

\vspace{3mm}  \noindent with $\textbf{u}_0  $ initial perturbation to $  \bold{\bar U} $  and  
$ (N\textbf{u})_{\textbf{u}_0} =\bold 0. $

\noindent Considering the linear  operator

 \begin{equation}
  \left \{
   \begin{array}{lll}
   L = \parallel a_{i,j}\parallel,  \quad(i,j=1,2...,n)\quad 
   \\
   \\
   a_{i,j} =  const. \in {\cal R} \,\,  \mbox{and independent from}\,\, t,      \end{array}
  \right.
\end{equation}

\vspace{3mm} 
\noindent the stability and instability  of   $ \bold{\bar U} $ is called linear  if it is evaluated via the linear system

\vspace{3mm} 
\noindent
\begin{equation} \label{2}
\frac{d \textbf{u}}{dt} = L \textbf{u} , \quad  t\geq 0, \quad \textbf{u}(0)= \textbf{u}_0
\end{equation}

\vspace{3mm} 
\noindent neglecting the nonlinear contribution  $N\textbf{u}.  $  


 In this regard, some theorems can be provided. 


\begin{theorem} 

If 

\begin{equation}  \label{7}
 det \,(a_{i,j}- \lambda \,\,\delta_{i,j} )  =0  \qquad \delta_{i,j} = \mbox { Kronecker numbers}
\end{equation}

\vspace{3mm}\noindent is the spectral equation whose  eigenvalues of the   $ n \,\mbox x\,n $ matrix $ ||a_{i,j} || $ are   $ \lambda_i (i =
1, 2,3...,n),$  and if and only if all the eigenvalues have negative real parts, then  $  \bold u=\bold 0  $  is \emph{linearly}
globally attractive and asymptotically exponentially stable. Otherwise, if there exists at least  an  eigenvalue with positive real part, then  $ \bold u =\bold 0  $ is 
unstable. \hbox{}\hfill$\square$\hspace{2.82mm}
\end{theorem}

 
 In addition, as proved in \cite{R1}, for a system formed by three equations such as the FHR  model,  the spectrum equation  (\ref{7}) of $ L  $  is reduced to the following expression:

\begin{equation}  \label{spe1}
 {\cal P}(\lambda) = \lambda ^3-I_1 \lambda^2 + I_2 \lambda -I_3  
 \end{equation}

\noindent where 

\begin{equation} 	\label{I}
I_{1} = a_{11}+a_{22}+a_{33}; \quad \quad I_{3} =  \mbox{det}  \,\, \parallel a_{i,j} \parallel,
\end{equation}

\begin{equation} \label{II}
I_{2} = \begin{vmatrix} a_{11} & a_{12} \\ a_{21} & a_{22} \end{vmatrix}+\begin{vmatrix} a_{11} & a_{13} \\ a_{31} & a_{33}\end{vmatrix}+\begin{vmatrix} a_{22} & a_{23} \\ a_{32} & a_{33}\end{vmatrix}, 
\end{equation}

\vspace{3mm} \noindent   represent the  invariants of $ L $ whose spectrum is the set   $ \sigma =\{ \lambda_1, \lambda_2, \lambda _3\} $  of its eigenvalues.
 Moreover, connected to the invariants $ I_i  (i=1,2,3), $ we can introduce the quantities:

\vspace{3mm} \noindent \begin{equation} \label{A}
 A_1=- \mbox{trace of}  \,L\,\,= - (\lambda_1 +\lambda_2 +\lambda_3 ) = - I_1; 
 \end{equation}

\begin{equation}\label{AA}
A_2 = \begin{vmatrix} a_{11} & a_{12} \\ a_{21} & a_{22} \end{vmatrix}+\begin{vmatrix} a_{11} & a_{13} \\ a_{31} & a_{33}\end{vmatrix}+\begin{vmatrix} a_{22} & a_{23} \\ a_{32} & a_{33}\end{vmatrix}  = \lambda_1 (\lambda_2 +\lambda_3 )+\lambda_2 \lambda_3 =I_2
\end{equation}

\noindent and 
 
 \begin{equation} \label{AAA}
  A_3= - \mbox{det of }\,\, L  \,=\,\,- \, \lambda_1\lambda_2\lambda_3 =- I_3
 \end{equation}

\vspace{3mm}\noindent  and,  according to \cite{R3}, the following  Lienard-Chipart criterion  holds:


\begin{theorem} \label{thaa}

  If  and only if

\begin{equation}  \label{aaa}
A_k > 0,\quad  (k=1,2,3)  \quad \mbox{and } \,\,  A_0=  A_1\,A_2 -A_3 > 0,
\end{equation}

\vspace{3mm} \noindent all the eigenvalues have negative real part. In particular, each of the conditions:

 \begin{equation}  \label{aaas}
A_1 > 0,\quad A_2 > 0, \quad A_3 > 0 ,
\end{equation}

\vspace{3mm} \noindent is necessary for all the roots to have negative real parts. Otherwise some roots will have positive real parts. 
\hbox{}\hfill$\square$\hspace{2.82mm}
\end{theorem}


Moreover, taking into account  that the instability  can occur  only  via a zero eigenvalue ($ \lambda =0 \Leftrightarrow A_3=0) $ or via a pure imaginary eigenvalues,  $ \lambda_{1,2} =\pm i\omega  \,\,(i  \,\,\mbox{imaginary unit,} $ 
$\,\, \omega \in \Re^+ )$ such that ${\cal P}(i\omega,R)=0, $   the onset of instability will be defined either as steady bifurcation or Hopf  bifurcation depending on wether  the
instability occurs through a steady or oscillatory state  \cite{R1}.

When the problem at issue depends on a positive parameter $ R, 
  $  let denote  by $ R_{c_k} $ the  lowest roots of value of $ R  $ such that  $ A_k(R)=0 $  for $ k=1,2,3. $ According to \cite{R3}, it is possible to introduce the  
\begin{equation}
 \mbox{"instability coefficient power" }(ICP)_k  \,\,\mbox{ of}\,\, A_k :\quad (ICP)_k =\frac{1}{R_{C_{k}}}
\end{equation}

\noindent and the following theorem   holds:
 
\begin{theorem}
Let $A_{\bar k}$ be the spectrum equation coefficient with the biggest ICP and let the critical  point  $ \bar C $ be linearly asymptotically stable at $R =\bar R = 0.$ Then, at the growing of $R$ from $R = 0$, the instability occurs at
$R= R_{C_{\bar k}}  $
and one has a steady bifurcation if $ \bar k =3, $ while an oscillatory bifurcation occurs at an $ R\in ]0, R_{C_{\bar k}}[ $  if $k < 3$.
\hbox{}\hfill$\square$\hspace{2.82mm}
\end{theorem}






\section{Mathematical model}

Let consider the FHZ system (\ref{22}) and assuming

\begin{equation}
 \eta = \beta \varepsilon; \qquad  \gamma = \delta d                                                      
\end{equation}

\noindent it results:

\begin{equation}
\label{244}
  \left \{
   \begin{array}{lll}
    \displaystyle{\frac{d \,U }{d \,t }} =\, - a U + U^2(a+1) - \frac{1}{k} U^3 
     \,-\, W\,\,+Y\,\, +I  \,  \quad \quad
     \\
     \\
      
\displaystyle{\frac{d \,W }{d \,t } }\, = \,- \eta   W +\varepsilon\, c +\varepsilon \,U  \quad \quad 
\\
\\
 
\displaystyle{\frac{d \,Y }{d \,t } }\, = \,-\delta \,U +\delta\,h - \gamma \,Y
   \end{array}
  \right.
\end{equation}

\vspace{3mm}\noindent  If  $ C=(\bar U ,\bar W, \bar Y) $ is an admissible critical point, considering:

 \begin{equation}
\label{11z}
  u= U-\bar U; \qquad 
 w=W -\bar W ;   
\qquad y= Y-\bar Y
\end{equation}

\vspace{3mm}

\noindent  as the perturbation vector, from  (\ref{244}) one obtains:

\begin{equation}
\label{11xxzz}
  \left \{
   \begin{array}{lll}
    \displaystyle{\frac{d \,u }{d \,t }} =\, - \frac{1}{k}\, u^3 -  \frac{3}{k}\, u^2  \,\bar U -  \frac{3}{k}\, u\, \bar U^2 - a u + (a+1) ( u^2 + 2 u \,\bar U ) - w + y   
     \, \quad \quad
     \\
     \\
      
\displaystyle{\frac{d \,w }{d \,t } }\, =   \varepsilon  u - \eta\,\,w \quad 
\\
\\
 
\displaystyle{\frac{d \,y }{d \,t } }\, =  - \delta  u  - \gamma \,y.
   \end{array}
  \right.
\end{equation}


\noindent  Linearizing  about $ C,  $ it results:

\begin{equation}
\label{11bvc}
  \left \{
   \begin{array}{lll}
    \displaystyle{\frac{d \,u }{d \,t }} =\,  u   \big[-   \frac{3}{k}\,\,\bar U^2   \,-  a + 2 \,(a+1) \,\, \bar U\,\big]\, - w\,+y 
     \,
          \\
     \\
      
\displaystyle{\frac{d \,w }{d \,t } }\, = \varepsilon u - \eta  \,w  
\\
\\
 
\displaystyle{\frac{d \,y }{d \,t } }\, =  -\delta  \, u -\gamma \, y. 
   \end{array}
  \right.
\end{equation}

\vspace{3mm} \noindent Denoting by

\begin{equation}
L= \begin{pmatrix}
\displaystyle - 3 \,\frac{1}{k}\, \bar U^2   \,+ 2 (a+1)  \bar U\, -  a \,\, & -1 \,\, & 1 \,\, \\
\\
\displaystyle \varepsilon \,\, & - \eta  \,\, & 0\,\, \\ \\
\displaystyle -\delta \,\, & 0\,\, & -  \gamma \, 
\end{pmatrix}
\end{equation}


\vspace{3mm}\noindent the linear operator,  according to  (\ref{I})-(\ref{II}), for $ k=3,  $ one has:

\vspace{3mm}\noindent \begin{equation}
\label{31}
  \left \{
   \begin{array}{lll}
    \displaystyle I_{1} = - \bigg[\,\bar U^2   \,- 2 \,(a+1) \, \bar U\, +  a + \eta\,+\gamma\, \bigg]
     \,
          \\
     \\
      
\displaystyle I_{2} = -\bigg( \eta\,\,+\gamma\bigg) \bigg[- \,\bar U^2   \,+ 2 (a+1)  \,\bar U\, -  a \bigg] + \varepsilon \,   + \delta  +\eta \gamma  
\\
\\
 
\displaystyle I_{3} = - \bigg\{\gamma \,\, \bigg[\eta  \bigg( \,\bar U^2   \,- 2 (a+1)  \,\bar U\,+  a  \,\bigg )+\varepsilon \bigg ]\,+\delta \eta \bigg \}    \end{array}
  \right.
\end{equation}

\vspace{3mm} \noindent   as the  invariants of $ L. $ Besides,  taking into account (\ref{A})-(\ref{AAA}) one deduces:

 \vspace{3mm}\noindent \begin{equation}
\label{31bbb}
  \left \{
   \begin{array}{lll}
    \displaystyle A_{1} =  \,\bar U^2   \,- 2 (a+1) \,\, \bar U\, +  a + \eta\,+\gamma\, 
     \,
          \\
     \\
      
\displaystyle A_{2} = -(\eta\,\,+\gamma) (- \,\bar U^2   \,+ 2 (a+1) \,\, \bar U\, -  a) + \varepsilon \,   + \delta  +\eta \gamma  
\\
\\
 
\displaystyle A_{3} = \gamma \,\,[\,\eta  ( \,\bar U^2   \,- 2 (a+1) \,\, \bar U\, +  a)+\varepsilon ]\,+\delta \eta,     
\end{array}
  \right.
\end{equation} 

\vspace{3mm}\noindent  and letting 

\begin{equation}  \label{uu}
 \displaystyle \Gamma= \bar U^2  \,- 2 (a+1)  \bar U\, +  a,  
\end{equation}


\noindent one obtains

 \begin{equation}
\label{31bbbx}
  \left \{
   \begin{array}{lll}
    \displaystyle A_{1} =  \Gamma + \eta\,+\gamma\, 
     \,
          \\
     \\
      
\displaystyle A_{2} = (\eta\,\,+\gamma) \, \Gamma + \varepsilon \,   + \delta  +\eta \gamma  
\\
\\
 
\displaystyle A_{3} = \gamma \,\eta \, \Gamma+ \gamma \varepsilon\,+\delta \eta 
\\
\\
\displaystyle A_{0} = A_1 A_2-A_3 = (\Gamma + \eta\,+\gamma\,)[(\eta\,\,+\gamma) \, \Gamma + \varepsilon \,   + \delta  +\eta \gamma ]- ( \gamma \,\eta \, \Gamma+ \gamma \varepsilon\,+\delta \eta). 
 \end{array}
  \right.
\end{equation}

 
 \section{Hopf bifurcations driven by $ \bar U $}

The FHR system depends on several parameters, and according to  each coefficient, various  Hopf bifurcations conditions can be obtained.  

In order to study  Hopf bifurcations driven by  critical point $ \bar U $, the attention is focused on

 \begin{center}
$ \displaystyle \Gamma= \bar U^2  \,- 2 (a+1)  \bar U\, +  a $
\end{center}

\vspace{3mm} \noindent already introduced in (\ref{uu}), and the following theorem  for linear stability can be proved:
 
 \begin{theorem} \label{th2}
 Let  $ \bar C= (\bar U, \bar W, \bar Y) $ be an admissible critical point and let assume  constants $ (\varepsilon,\, \delta,\, d,\,  \beta),  $  be positive. Then,    whatever the value of variable $ a \in R $ may be,  if

   \begin{equation} \label{asd}
  \left \{
  \begin{array}{lll}
   \displaystyle  \bar U\leq -\sqrt{a^2+a+1}\, +a+1
      \\ 
      \mbox{or}\\
        \displaystyle  \bar U\geq \sqrt{a^2+a+1}\, +a+1,
 \end{array}
  \right.
\end{equation}

\vspace{3mm}\noindent then the critical point $ \bar C $ is linearly, globally attractive and asymptotically exponentially stable.
 \end{theorem}

\begin{proof} Condition  (\ref{asd})   ensures that $ \Gamma \geq 0,  $  and it is possible to prove that the positiveness of the FHR system's constants implies  that   $ A_k , \, (k=0,1,2,3), $  determined in   (\ref{31bbbx}), are all non-negative. Moreover, they are  increasing functions of $ \Gamma. $

This ensures  that conditions (\ref{aaa})   of theorem \ref{thaa}  state, and hence theorem holds.

\end{proof}

When conditions (\ref{asd}) are not satisfied, i.e   the critical point $ \bar U  $ is such that the following inequality:

   \begin{equation} \label{asd1}
  \begin{array}{lll}
    \displaystyle -\sqrt{a^2+a+1}\, +a+1 < \bar U < \sqrt{a^2+a+1}\, +a+1   \quad   \forall a \in {\cal R }
 \end{array}
\end{equation} 


\noindent holds, then it results $ \Gamma <0 $ and in this case it is possible to introduce a positive parameter $ R $  as \emph{"bifurcation parameter"}. Indeed if we  let:

\begin{equation}
\label{111poi}
  \left \{
   \begin{array}{lll}
    \displaystyle  R=- \Gamma= -[ \bar U^2  \,- 2 (a+1)  \bar U\, +  a ]>0
          \\
     \\
 \displaystyle  c_1 = \eta  + \gamma
 \\
\\
 
\displaystyle  c_2  =\,\, \frac{ \varepsilon \,+  \delta  + \eta \,\gamma  }{\eta \,+\,\gamma }  
\\ \\ 
\displaystyle  c_3 = \, \frac{\gamma \varepsilon  +\eta \delta  }{\eta \gamma   } = \frac{d+\beta}{\beta d}
   \end{array}
  \right.
\end{equation}


\noindent it results:

 \begin{equation}
\label{31bbbxx}
  \left \{
   \begin{array}{lll}
    \displaystyle A_{1} =  -R  + c_1 
     \,
          \\
     \\
      
\displaystyle A_{2} = - c_1 \, R +   c_1 \,\,c_2  
\\
\\
 
\displaystyle A_{3} = -\gamma \,\eta \,R +  c_3 \,\,\gamma \,\eta \,
 \end{array}
  \right.
\end{equation}


\noindent with:

\begin{equation}
\label{x}
      \displaystyle A_{1} = 0 \Leftrightarrow R = c_1 ;
     \,\quad  A_{2} =0 \Leftrightarrow  \,R =c_2;   \quad A_{3} = 0 \Leftrightarrow R =  c_3. 
\end{equation}


\noindent So that,  denoting by $ R_{c_k} $ the  lowest roots of value of $ R  $ such that  $ A_k=0 $  for $ k=1,2,3, $   one has:

\begin{equation}
\displaystyle R_{c_k} = \min_{(\beta,d, \varepsilon, \delta)\in R^+ }   c_k \quad k=1,2,3 
\end{equation}

\noindent and the following theorem holds:

\begin{theorem}  \label{th55}
In the hypothesis (\ref{asd1}),  let $  R= -\Gamma  = -[ \bar U^2  \,- 2 (a+1)  \bar U\, +  a ]>0 $ and let constants $ (\varepsilon,\, \delta,\, d,\,\beta ),  $  be positive.

Then, at the growing of $ R  $ from $ R=0, $  conditions

\begin{equation}\label{i}
\displaystyle \eta + \gamma <  \frac{ \varepsilon \,+  \delta  + \eta \,\gamma  }{\eta \,+\,\gamma }; \qquad  \eta + \gamma <\frac{d+\beta}{\beta d} ;  
 \end{equation}

 
\vspace{3mm} \noindent ensure that a\emph{ simple  oscillatory bifurcation} occurs at a $ \bar R \in ]0, R_{C_1}[ ,$ with a frequency $\displaystyle  \frac{\varphi }{2\pi}\, $ where $ \displaystyle \varphi^2 = \frac{A_3(\bar R)}{ A_1(\bar R)} = A_2 (\bar R). $

 \vspace{3mm} \noindent If, in particular

\begin{equation}\label{i*}
\displaystyle \eta + \gamma =  \frac{ \varepsilon \,+  \delta  + \eta \,\gamma  }{\eta \,+\,\gamma }; \qquad  \eta + \gamma <\frac{d+\beta}{\beta d} 
 \end{equation}

 \vspace{3mm}
 \noindent a simple oscillatory bifurcations occurs at a $ \bar R  \in ]0, R_{C_1}=R_{C_2} [. $

 
 \vspace{3mm} \noindent  Otherwise, if 

  \begin{equation} \label{ii}
\eta + \gamma = \frac{d+\beta}{\beta d} <  \frac{ \varepsilon \,+  \delta  + \eta \,\gamma  }{\eta \,+\,\gamma } 
\end{equation}


 \vspace{3mm}  \noindent a steady+oscillatory bifurcation appears with a frequency given by $ \varphi= (2\pi)(\sqrt{A_2})_{R_{c_1}}. $


\vspace{3mm}  \noindent Moreover, if 

\begin{equation}\label{iv}
 \displaystyle  \frac{ \varepsilon \,+  \delta  + \eta \,\gamma  }{\eta \,+\,\gamma } < \eta + \gamma < \frac{d+\beta}{\beta d} 
\end{equation}

\vspace{3mm}  \noindent { a simple  oscillatory bifurcation} occurs at a $ \bar R \in ]0, R_{C_2}[.$

\end{theorem}

\begin{proof} When  $  R=-\Gamma =0,  $ it results: 
\[    \displaystyle A_{0}]_{\Gamma=0} = A_1 A_2-A_3]_{\Gamma=0}=
( \eta\,+\gamma\,)( \varepsilon \,   + \delta  +\eta \gamma )- (  \gamma \varepsilon\,+\delta \eta) =  \eta\,\varepsilon  +\gamma\, \delta+ ( \eta\,+\gamma\,) \eta \gamma >0 \]

\noindent with $ A_k  >0\, (k=1,2,3) $. So, the critical point  is linearly, asymptotically stable for $ R=0. $  


\noindent Besides, when  inequalities (\ref{i}) hold, it means that  \begin{center}
$  R_{C_1} < R_{C_2};  \qquad   R_{C_1} < R_{C_3};  $ 
\end{center}


\noindent i.e.  $A_{1}$ is  the spectrum equation coefficient with the biggest instability coefficient power,  so that  at $  R= R_{C_1}, $  it results:

 \begin{center}
 $ A_1=0, \quad  A_3 = \gamma \eta \,( -c_1 + c_3) >0;   \qquad A_0=A_1 A_2-A_3 <0  $
 \end{center}
 

\noindent and hence, in view of the continuity of $ A_1 A_2-A_3, $ there exists a $  \bar  R  \in ]0, R_{C_1}[ $  such that

  \noindent 
  \begin{center}
 $ A_1 (\bar R) A_2(\bar R) -A_3  (\bar R)=0,  $
 \end{center}

 \vspace{3mm} 
 \noindent being $ \bar R  $   the lowest root of $ A_1 A_2 = A_3$   in $ ]0, R_{C_1}[ $ and it results

\vspace{3mm} \noindent 
\begin{equation} \label{d}
\displaystyle P(i \varphi, \bar R)=0\Leftrightarrow \displaystyle  [\lambda^3+A_1 \lambda ^2 + A_2  \lambda + A_3]_{i\varphi } =0
\end{equation}  


\noindent  and hence 

\begin{equation} \label{d}
\displaystyle  -  i\varphi^3-A_1(\bar R) \, \varphi ^2 + i \,A_2(\bar R) \,\varphi+ A_3(\bar R) =0 
\end{equation}  


\noindent with

 \begin{equation} \label{aa}
\displaystyle  \varphi^2 = \frac{A_3(\bar R)}{ A_1(\bar R)} = A_2 (\bar R).
\end{equation}

\vspace{3mm} \noindent
 Besides, conditions  (\ref{i*})  imply that  $  R=R_{C_1} = R_{C_2} < R_{C_3} $ that means $ A_1=A_2=0$ and

\noindent   \begin{center}
$ A_3(R_{C_1})= -\gamma \,\eta \, c_1+  c_3 \,\,\gamma \,\eta \,  >0 $ 
\end{center}
 

 \noindent Consequently, the spectrum equation is reduced to:

\vspace{3mm} 

\noindent  
\begin{equation}
\lambda^3 + A_3 = (\lambda + A_3^{1/3}) \,( \lambda ^2 -  \lambda\, A_3^{1/3}  + A_3^{2/3})=0 
\end{equation}


\noindent  and hence 
 
 \[ \lambda_1= - A_3^{1/3} \quad \lambda_{2,3}= ( 1\pm \,i\,\sqrt{3}
)  A_3^{1/3}/2. \, \]

\vspace{3mm} \noindent This means that a simple oscillatory bifurcation occurs at a $ \bar R \in ]0, R_{c_1}= R_{c_2}[. $

\vspace{3mm} \noindent Instead, when (\ref{ii}) holds,  $ R_{c_1}= R_{c_3}<R_{c_2} $ ; and hence one obtains $ A_1= A_3=0.  $ So, from the spectrum equation it results:

 \vspace{3mm} 
 
 \noindent 
\begin{equation}
P(i\varphi)=0\Leftrightarrow [\lambda ( \lambda ^2 +   A_2)]_{i\varphi}=0 \Leftrightarrow \lambda=0 ; \varphi= (\sqrt{A_2})_{R_{c_1}} =    \sqrt{c_1   \,(c_2 - c_1 ) } 
\end{equation}

 \vspace{3mm} 
 \noindent and a steady ($ \lambda=0 $) $ + $ oscillatory bifurcation  of frequency $ \varphi/\pi $ with $ \varphi= (\sqrt{A_2})_{R_{c_1}} $  occurs.

 
 \noindent  
Analogous results can be obtained  if we suppose $ R_{c_2} $  to be the biggest ICP and hence (\ref{iv}) is proved, too.
\end{proof}

\section { Hopf bifurcations driven by $ -\eta = - \varepsilon \beta  >0 $ }

The previous bifurcation criterion required that $ \Gamma \leq 0. $ In the present section, we prove  that, by  choosing  
   $ \eta= \varepsilon \beta $ as bifurcating parameter and letting  $ \eta \leq 0, $ the Hopf bifurcation can arise with  $\Gamma \geq 0.$

Indeed,  the following theorem states:

\begin{theorem}
Let consider a critical point $ \bar C $ such that:

  \begin{equation} \label{asdddn}
     \displaystyle  \bar U\leq -\sqrt{a^2+a+1}\, +a+1 \,\,\,\mbox{or}  \,\,\, \bar U\geq \sqrt{a^2+a+1}\, +a+1
\end{equation}  

\vspace{3mm}\noindent and  let  constants $ (\varepsilon,\, \delta,\, d ),  $  be positive. Assuming $ R = - \eta= - \varepsilon \,\beta  >0, $  then, at the growing of $ R  $ from $ R=0, $  conditions

\begin{equation}\label{in}
\displaystyle \Gamma + \gamma  \leq  \frac{ \gamma  \Gamma +\varepsilon \,+  \delta     }{\Gamma\,+\gamma }  ; \qquad  \Gamma + \gamma <\frac{\gamma \varepsilon}{ \gamma \Gamma +\delta}
\end{equation}

\vspace{3mm}\noindent  ensure that a\emph{ simple  oscillatory bifurcation} occurs at a $ \bar R \in ]0, R_{C_1}[ ,$ while if

  \begin{equation} \label{iin}
\Gamma + \gamma  = \frac{\gamma \varepsilon}{ \gamma \Gamma +\delta} <  \frac{ \gamma \Gamma +\varepsilon \,+  \delta     }{\Gamma\,+\gamma } 
\end{equation}

\vspace{3mm} \noindent a steady+oscillatory bifurcation appears.

\vspace{3mm}\noindent Moreover, if 

\begin{equation}\label{ivn}
 \displaystyle  \frac{ \gamma \Gamma +\varepsilon \,+  \delta     }{\Gamma\,+\gamma }< \eta + \Gamma < \frac{\gamma \varepsilon}{ \gamma \Gamma +\delta},
 \end{equation}

 \vspace{3mm}\noindent  a{ simple  oscillatory bifurcation} occurs at a $ \bar R \in ]0, R_{C_2}[.$
 
\end{theorem}

\begin{proof}Condition (\ref{asdddn}) ensures that $ \Gamma, $ defined in (\ref{uu}), is positive.  Moreover,   since (\ref{31bbbx}), when  $ \eta =0, $ it results  $ A_k>0 \,\, (k=1,2,3) $ and 

\[ \displaystyle   A_0]_{\eta=0}  =   \displaystyle  A_1 A_2-A_3]_{\eta=0} = ( \Gamma \,+\gamma\,)( \Gamma \,\gamma +\varepsilon \,   + \delta )-   \gamma \varepsilon\, >0.  \]

\vspace{3mm}\noindent So,  the  critical point $ \bar C= (\bar U, \bar W, \bar Y) $  is linearly, asymptotically  stable for $ R = \bar R =0 $. In addition, denoting by 

 \begin{equation}
\label{111poi}
   \displaystyle  c_1 = \Gamma  + \gamma; \quad c_2  =\,\, \frac{ \gamma \Gamma +\varepsilon \,+  \delta     }{\Gamma\,+\gamma };  \quad   c_3 = \frac{\gamma \varepsilon}{ \gamma \Gamma +\delta}
  \end{equation}

 \vspace{3mm}\noindent from (\ref{31bbbx}) it results

\begin{equation}
\label{x}
      \displaystyle A_{1} = 0 \Leftrightarrow R = c_1 ;
     \,\quad  A_{2} =0 \Leftrightarrow  \,R =c_2;   \quad A_{3} = 0 \Leftrightarrow R =  c_3, 
\end{equation}

 \vspace{3mm}\noindent and so, by retracing the analysis set forth in the previous bifurcation cases, this theorem can also be proved.  
\end{proof}

\section{Remarks and discussion}
 As it  is well known, the phenomenon related to  Hopf bifurcations is of great importance and  it is widely studied.
In this paper, the   FHR  model (\ref{22}) considered  also depends on the variable $a$ generally not present in the bifurcations studies and it generalizes the FHR  system (\ref{21p}), which, on the contrary, is more often considered in literature.

Moreover, the results obtained [see, f.i. Theorems \ref{th2}-\ref{th55} and condition (\ref{asd1})] do not require any assumptions for the real variable $a$ and this implies that the analysis can certainly be directed to a wider set of physical cases.

Furthermore, the equivalence that such a mathematical model creates between biological problems and  superconducting processes of Josephson junctions or viscoelasticity,
 suggests that the analysis of such models is reflected in a large number of realistic mathematical models.

In this paper the onset of Hopf bifurcations,  driven by specific parameters, is considered. In particular   an analysis  on the onset of steady and  oscillatory bifurcations has been performed  driven by both an admissible  critical point $ \bar U  $  and a coefficient characterized the FHR system.

Looking forward, in order to  obtain a more comprehensive view of the stability and instability of critical points, the analysis can be extended to evaluate Hopf bifurcations driven by all other coefficients that characterize the FHR system.  Moreover, it will be possible to determine explicit critical points at particular values of the FHR system variables and also evaluate the explicit value of the bifurcation parameters $R.$



\vspace{6mm}

\textbf{{Acknowledgments}}

This paper has been performed under the auspices of the National Group of Mathematical Physics GNFM-INdAM.The author is grateful to the anonymous reviewers for their comments and suggestions.

\end{document}